%% file: LFlow_full.tex
\documentclass[fullpage,11pt,a4paper]{article}
\usepackage[english]{babel}
\usepackage{graphicx}
\usepackage{fullpage}

\usepackage{amsmath}        
\usepackage{amsfonts}       
\usepackage{amsthm}         

\usepackage{algorithm}
\usepackage{algorithmic}
\usepackage[T1]{fontenc}

\usepackage{hyperref}

\usepackage{xargs} 
\usepackage[colorinlistoftodos,prependcaption,textsize=tiny]{todonotes}
\newcommand{\Riha}{\v R\'\i ha}

\include{makra_full}

\begin{document}

\title{\bf On Polynomial-Time Combinatorial Algorithms for Maximum $L$-Bounded Flow }
\author{
	Kate\v{r}ina Altmanov\' a, Petr Kolman, Jan Voborn\' \i k \\
	Department of Applied Mathematics\\
	Faculty of Mathematics and Physics\\
	Charles University, Prague
}

\maketitle

\begin{abstract}
Given a graph $G=(V,E)$ with two distinguished vertices $s,t\in V$ and an
integer $L$, an {\em $L$-bounded flow} is a flow between $s$ and $t$ that can
be decomposed into paths of length at most $L$.  In the {\em maximum
$L$-bounded flow problem} the task is to find a maximum $L$-bounded flow
between a given pair of vertices in the input graph.  

The problem can
be solved in polynomial time using linear programming.  However, as far as we
know, no polynomial-time combinatorial algorithm for the $L$-bounded flow is
known. 
The only attempt, that we are aware of, to describe a combinatorial algorithm
for the maximum $L$-bounded flow problem was done by Koubek and \v R\'\i ha in
1981.  Unfortunately, their paper contains substantional flaws and the
algorithm does not work; in the first part of this paper, we describe these
problems.  

In the second part of this paper we describe a combinatorial
algorithm based on the exponential length method that finds 
a $(1+\epsilon)$-approximation of the maximum $L$-bounded flow in 
time $\O(\eps^{-2}m^2L\log L)$ where $m$ is the number of edges in the
graph. 
Moreover, we show that this approach works even for the 
NP-hard generalization of the maximum $L$-bounded flow problem in which
each edge has a length.  
\end{abstract}

\section{Introduction}
Given a graph $G=(V,E)$ with two distinguished vertices $s,t\in V$ and an
integer $L$, an {\em $L$-bounded flow} is a flow between $s$ and $t$ that can
be decomposed into paths of length at most $L$.  In the {\em maximum
$L$-bounded flow problem} the task is to find a maximum $L$-bounded flow
between a given pair of vertices in the input graph.  
The $L$-bounded flow was first studied, as far as we know, in 1971 by Ad\'amek
and Koubek~\cite{adamek1971remarks}. In connection with telecommunication
networks, $L$-bounded flows in networks with unit edge lengths have been widely
studied and are known as \emph{hop-constrained} flows~\cite{BleyNeto}.

For networks with unit edge lengths (or, more generally, with polynomially
bounded edge lengths, with respect to the number of vertices), the problem can
be solved in polynomial time using linear programming.  Linear programming is a
very general tool that does not make use of special properties of the problem
at hand. This often leaves space for superior combinatorial algorithms that do
exploit the structure of the problem.  For example, maximum flow, matching,
minimum spanning tree or shortest path problems can all be described as linear
programs but there are many algorithms that outperform general linear
programming approaches.  However, as far as we know, no polynomial-time
combinatorial algorithm\footnote{Combinatorial in the sense that it does not
explicitly use linear programming methods or methods from linear algebra or
convex geometry.} for the $L$-bounded flow is known.

\subsection{Related results}

For clarity we review the definitions of a few more terms that are used in this
paper.  
A {\em network} is a quintuple $G = (X, R, c, s, t) $, where $G = (X, R)$ is
a directed graph, $X$ denotes the set of vertices, $R$ the set of edges, $c$ is
the edge capacity function $c: R \rightarrow \mathbb{R}^+$, $s$ and $t$ are two
distinguished vertices called the source and the sink. We use $m$ and $n$ to denote the
number of edges and the number of vertices, respectively, in the network $G$, that is, 
$m=|R|$ and $n=|X|$.
Given an $L$-bounded flow $f$, we denote by $|f|$ the size of the flow, and for
an edge $e\in R$, we denote by $f(e)$ the total amount of flow $f$ through the
edge $e$.

An {\em $L$-bounded flow problem with edge lengths} is a generalization of the
$L$-bounded flow problem: each edge has also an integer length and the length
of a path is computed not with respect to the number of edges on it but with
respect the sum of lengths of edges on it.

Given a network $G$ and an integer parameter~$L$, an {\em $L$-bounded cut} is a
subset $C$ of edges $R$ in $G$ such that there is no path from $s$ to $t$  of
length at most $L$ in the network ${G = (X, R \setminus C, c, s, t)}$.  The
objective is to find an $L$-bounded cut of minimum size.  We sometimes
abbreviate the phrase \kbc to {\kcw} and, similarly, we abbreviate the phrase
\kbf to \kfw . 

Although the problems of finding an \kf and an \kc are easy to define
and they have been studied since the 1970's, still some fundamental open
problems remain unsolved. Here we briefly survey the main known results.

\paragraph{L-bounded flows}

As far as we know, the $L$-bounded flow was first considered in 1971 by Ad\'amek and
Koubek~\cite{adamek1971remarks}.  They published a paper introducing the \kbfw
s and cuts and describing some interesting properties of them. Among other
results, they show that, in contrast to the ordinary flows and cuts, the
duality between the maximum \kf and the minimum $L$-cut does not hold.

The maximum \kf can be computed in polynomial time using linear
programming~\cite{baier2010length,kolman2006improved,baier2010length,mahjoub2010max}.
The only attempt, that we are aware of, to describe a combinatorial algorithm
for the maximum $L$-bounded flow problem was done by Koubek and \v R\'\i ha in
1981~\cite{koubek1981maximum}.  
The authors say the algorithm finds a
maximum \kf in time $ O(m \cdot |I|^{2}\cdot S/\psi(G)) $, where $I$ denotes
the set of paths in the constructed $L$-flow, $S$~is the size of the maximum
\kfw, and $ \psi(G) = \min( |c(e) - c(g)|: c(e) \neq c(g), e, g \in R \cup
\{e'\} ) $, where $ c(e') = 0 $. 
Unfortunately, their paper contains
substantional flaws and the algorithm does not work as we show in the first
part of this paper. Thus, it is a challenging problem to find a polynomial time
combinatorial algorithm for the maximum $L$-bounded flow. 

Surprisingly, the maximum $L$-bounded flow problem with edge lengths is
NP-hard~\cite{baier2010length} even in outer-planar graphs.
Baier~\cite{baier2003flows} describes a FPTAS for the maximum $L$-bounded flow
with edge lengths that is based on the ellipsoid algorithm.  He also shows that
the problem of finding a decomposition of a given $L$-bounded flow into paths
of length at most $L$ is NP-hard, again even if the graph is outer-planar.

A related problem is that of $L$-bounded disjoint paths: the task is to find
the maximum number of vertex or edge disjoint paths, between a given pair
of vertices, each of length at most $L$. The vertex version of the problem is 
known to be solvable in polynomial time for $L\leq 4$ and NP-hard for 
$L\geq 5$~\cite{itai1982complexity}, and the edge version is solvable 
in polynomial time for $L\leq 5$ and NP-hard for $L\geq 6$~\cite{Bley:03}.

\paragraph{L-bounded cuts}

The \kbc problem is NP-hard
\cite{Schieber:1995:CFM:241577}.  Baier et al.~\cite{baier2010length} show that
it is NP-hard to approximate it by a factor of
$1.377$ for $L\geq 5$ in the case of the vertex $L$-cut, and for $L\geq
4$ in the case of the edge $L$-cut. Assuming the Unique Games
Conjecture, Lee at al.~\cite{Lee2017ImprovedHF} proved that the minimum \kbcw \
problem is NP-hard to approximate within any constant factor. For planar
graphs, the problem is known to be NP-hard~\cite{FHNN:15,ZFMN:17}, too.

The best approximations that we are aware of are by Baier et
al.~\cite{baier2010length}: they describe an algorithm with an
$\mathcal{O}(\min\{L, n/L\}) \subseteq \mathcal{O}(\sqrt{n})$-approximation for
the $L$-bounded vertex cut, and 
$\mathcal{O}(\min\{L, n^{2}/L^{2}, \sqrt{m}\}) \subseteq
\mathcal{O}(n^{2/3})$-approximation for the $L$-bounded edge cut.
The approximation factors are closely related with the cut-flow gaps:
there are instances where the minimum edge 
$L$-cut (vertex $L$-cut) is $\Theta(n^{2/3})$-times
($\Theta(\sqrt{n})$-times) bigger than the maximum \kfw~\cite{baier2010length}. 
For the vertex version of the problem, there is a $\tau$-approximation
algorithm for graphs of treewidth $\tau$~\cite{Kolman2018OnAE}.

The \kbc was also studied from the perspective of parameterized complexity.
It is fixed parameter tractable (FPT) with respect to the treewidth of the
underlying graph~\cite{Dvork2015ParametrizedCO,Kolman2018OnAE}.
Golovach and Thilikos~\cite{Golovach2009PathsOB} consider several
parameterizations and show FPT-algorithms for many variants of the problem
(directed/undirected graphs, edge/vertex cuts). 
On planar graphs, it is FPT with respect to the length bound 
$L$~\cite{Kolman2018OnAE}.

The \kbc appears in the literature also as the short paths interdiction
problem~\cite{Bazgan2018AMF}, \cite{Kolman2018OnAE}, \cite{Lee2017ImprovedHF}
or as the most vital edges for shortest paths~\cite{Bazgan2018AMF}.  

\subsection{Our contributions}
In the first part of the paper, we show that the combinatorial algorithm by
Koubek and {\Riha}~\cite{koubek1981maximum} for the maximum $L$-bounded flow is
not correct.

In the second part of the paper we describe an iterative combinatorial 
algorithm, based on the exponential length method,
that finds a $(1+\epsilon)$-approximation of the maximum $L$-bounded
flow in time $\O(\eps^{-2}m^2L\log L)$
; that is, we describe a fully polynomial approximation scheme (FPTAS)
for the problem. 

Moreover, we show that this approach works even for the
NP-hard generalization of the maximum $L$-bounded flow problem in which
each edge has a length.
This approach is more efficient than the FPTAS based on the ellipsoid
method~\cite{baier2003flows}.  

Our result is not surprising (e.g., Baier~\cite{baier2003flows} 
mentions the possibility, without giving the details, to use the exponential
length method to obtain a FPTAS for the problem); however, considering the
absence of other polynomial time algorithms for the problem that are not based
on the general LP algorithms, despite of the effort to find some, we regard it
as a meaningful contribution.  The paper is based on the results in the
bachelor's thesis of Kate\v{r}ina Altmanov\' a~\cite{Altmanova} and in the
master's thesis of Jan Voborn\' \i k~\cite{Vobornik}.

\section{The algorithm of Koubek and~\v{R}\'{i}ha}

\subsection{Increasing an $L$-bounded flow}

Before describing the problem with the algorithm by Koubek and 
{\Riha}~\cite{koubek1981maximum}, we informally describe the purpose and the
main attributes of {\em an increasing $L$-system}, a key structure used in the
algorithm. 

Consider a network $G=(X, R, c, s, t)$ and an arbitrary \kbf  $f$ from $s$ to $t$ in $G$,
together with its decomposition into paths of length at most $L$ (say paths
$p_1, p_2, \ldots $ carrying $r_1, r_2, \ldots $ units of flow, resp.)
that is not a~maximum \kbfw . Given $G$ and $f$, Koubek
and \v{R}\'{i}ha \cite{koubek1981maximum} build a labeled oriented tree
$T=(V,E,v_{0},LABV,LABE)$  where $V$ is the set of nodes, $E$ is the set of
edges, $ v_{0} $ is the root, $LABV$ is a vertex labelling and $LABE$ is an
edge labeling.  The tree is called {\em an increasing $L$ system with respect
to~$f$}.

There are four types of the nodes of the tree $T$; to explain the error in the paper,
it is sufficient to deal with three of them: $1$-son, $3$-son, $4$-son. 
With (almost) each node $u$ in $T$, are associated two consecutive paths in $G$:
the first one, denoted by $q(u)$, contains only
edges that are not used by the current \kffw, and the second one, denoted by
$\bar q(u)$, coincides with a subpath of some path from the current \kffw.
 \ (Fig.~\ref{fig:concatenation}).
\begin{figure}[h] 
	\centering
        \includegraphics[scale =1]{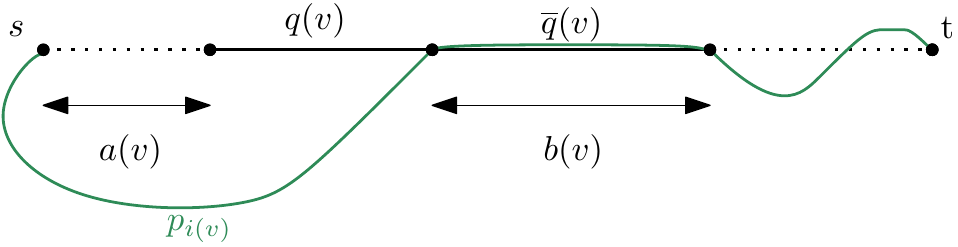}
        \caption{The concatenation of the section $ q(v) $ to $ \overline{q}(v) $.
		\label{fig:concatenation} }
\end{figure}
The tree $T$ encodes a combination of these paths with paths in $f$ and this
combination is supposed to yield  a larger $L$-flow than the \kffw.

The label of a vertex $v$ in the tree $T$, denoted by LABV in the original
paper, and the label of the edge $e$ connecting $v$ to its immediate ancestor,
if there is one, denoted by LABE, are of the following form:

\hspace{2cm}
\begin{tabular}{|l|l|l|}
\hline
 & LABV & LABE\\ \hline
$1$-son & ${(q(v),i(v),a(v),b(v))}$  & none \\ \hline
$3$-son or $4$-son & ${(q(v),i(v),a(v),b(v))}$  & $(h(e),j(e),d(e),o(e))$ \\ \hline
\end{tabular}
\\
where
\begin{itemize}
\setlength\itemsep{0mm}
\item $q(v)$ is a path in $G$ that is edge disjoint with every path in the $L$-flow $f$,
\item $i(v), j(v)$ are indices of paths in the $L$-flow $f$,
\item $a(v), b(v), d(e)$ are positive integers (distances),
\item $o(e)$ is a positive integer, if $v$ is a $3$-son, and $o(v)$ is a pointer to 
a $3$-son, if $v$ is a $4$-son,
\item $h(e)$ is a subset of edges in $G$.
\end{itemize}

As for every node $v$ in the tree (except for the root) there is a unique edge
$e$ connecting it to its parent, Koubek and {\Riha} often refer to the label of 
the edge $e$, and to its attributes, by the name of the vertex $v$, e.g., they
write $h(v)$ instead of $h(e)$; we shall use the same convention.

The tree $T$ is supposed to describe an $L$-flow $f'$ derived from $f$. 
In
particular, each path $q(v)$ and $\bar q(v)$ is a subpath of a new path
between $s$ and $t$ of length at most $L$. 
Very roughly speaking, the
attributes $a(v)$ and $d(v)$ store information about the distance of the path
segments $q(v)$ and $\bar q(v)$ from $s$ along the paths used in the
new $L$-flow $f'$, the attribute $i(v)$ specifies the index of a path from $f$
s.t. $\bar q(v)$ is a subpath of $p_{i(v)}$, and the attributes $b(v)$ and $o(v)$, resp., specify
the number of edges along which the paths $p_{i(v)}$ and $p_{j(v)}$ are being
followed by some of the new paths.

Consider a node $w$ in the tree $T$ such that at least one edge in $\bar q(w)$,
say an edge $e$, is saturated in the $L$-flow $f$ (i.e., $f(e)=c(e)$). 
In this case, the properties of the tree
$T$ enforce that the node $w$ has at least one $3$-son $u$ whose responsibility
is to desaturate the edge $e$ by diverting one of the paths 
that use $e$ in $f$ along a new route; the attribute $j(u)$ specifies the index of the path from $f$
that is being diverted by the $3$-son $u$ of $w$ (Fig.~\ref{fig:3-son}), 
\begin{figure}[h] 
	\centering
	\includegraphics[scale =1]{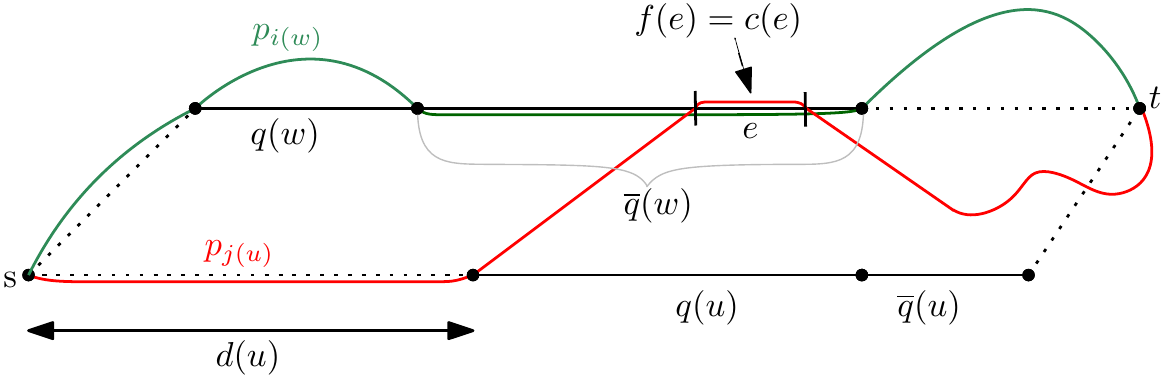} \label{obr:3syn}
	\caption{Desaturation of a saturated edge $e$ in a $\bar q(w)$ by a $3$-son $u$.
		\label{fig:3-son}}
\end{figure}
and $h(w)$ specifies which saturated
edges from $\bar q(v)$ are desaturated by the son $u$ of $w$.

As the definition of the tree $T$ does not pose any requirements on the
disjointness of the $\bar q$-paths corresponding to different nodes of
$T$, it may happen that the paths $\bar q(w)$ and $\bar q(w')$ for two different
nodes $w$ and $w'$ of the tree $T$ overlap in a saturated edge $e$. In this
case, Koubek and {\Riha} allow an {\em exception} (our terminology) to the rule
described in previous paragraph: if one of the nodes $w$ and $w'$, say the node
$w$, has a $3$-son $u$ that desaturates $e$, the other node, the node $w'$, need not
have a $3$-son but it may have a $4$-son instead. The purpose of this $4$-son
is just to provide a pointer to the $3$-son $u$ of $w$ that takes care about
the desaturation of the edge $e$.

\subsection{Small mistakes and typos}\label{smallMistakes}

The paper is full of small mistakes and typos which change the meaning. 
Here we mention the most striking typo.
On the page 393 in the paper \cite{koubek1981maximum}, there is the rule 3b:
	\begin{center}
			If $v$ is a $1$-son \textbf{of} a $3$-son then $v$ has a $1$-son 
		if and only if \\ $(END(q(v)) + b(v)) \mod\ p_{i(v)} \neq t$. 
	\end{center}
where $v$ is note in the tree $T$, $t$ is a vertex in the graph $G$, 
$END(q(v))$ denotes the last vertex of the path $q(v)$, and for a path $p$,
a vertex $w$ on $p$ and an integer $k$, $w+k \mod p$ denotes the vertex on the 
path $p$ that is $k$ edges after $w$.
The correct reading of the above rule, with a significantly different meaning, is:
        \begin{center}
                        If $v$ is a $1$-son \textbf{or} a $3$-son, then $v$ has a $1$-son 
                if and only if \\ $(END(q(v)) + b(v)) \mod\ p_{i(v)} \neq t$. 
        \end{center}

The difficulty with the original version is that it does not guarantee that the
paths in new $L$-flow $f'$ terminate in the vertex $t$.


\subsection{The main error}\label{mainError}

We start by recalling a few definitions and lemmas from the 
original paper~\cite{koubek1981maximum};
for the definition of the increasing system (more than one page long)
we refer to~\cite{koubek1981maximum}. 

\begin{defn}[Definition 4.2 in \cite{koubek1981maximum}]\label{d4.2}
	Let $ T $ be an \iks with respect to an~\kf $ f = \{ (p_i,r_i): i \in I\} $ in a network $ G = (X, R, c, s, t) $. Given an edge $ u \in R $, we define:
	\begin{itemize}
		\item $ T_1(u) $ is the number of vertices $ x $ in the tree $ T $ such that $ u \in \overline{q}(x) $ and if there is a saturated edge $ v \in \overline{q}(x) $ then there is a $ 3 $-son $ y $ of $ x $ with $ v \in h(y) $, $ u \notin p_{j(y)} $. 
		\item $ T_2(u) $ is the number of vertices $ x $ in the tree $ T $ such that $ u \in q(x) $.
		\item $ T_3(u) $ is the number of vertices $ x $ which are $ 3 $-sons or $ 4 $-sons with $ u \in h(x) $. 
	\end{itemize}
	For $ i \in I $ we denote $ m_i = \sup \{T_3(u): u \in p_i \} $, $|T| = \min \{\frac{c(u)}{T_2(u)}: u\in R, f(u) = 0\} \cup \{\frac{c(u) - f(u)}{T_1(u)}: u \in R\} \cup \{\frac{r_i}{m_i}: i \in I\}$, where the expressions that are not defined are omitted.
\end{defn}

\begin{lemma}[Lemma 4.2 in \cite{koubek1981maximum}]\label{l4.2}
	If there is an \iks with respect to an \kf $f$, then there is an \kf $g$ 
with $ |g| = |f| + |T| $.
\end{lemma}

\begin{defn}[Definition 4.3 in \cite{koubek1981maximum}]\label{d4.3}
	Let $ \overline{R} = R \cup \{u'\} $, where $ u' \notin R $ and $ c(u') = 0 $. 
We put $ \psi(G) = \min( |c(u) - c(v)|: c(u) \neq c(v), u, v \in \overline{R} )$.
\end{defn}

\begin{figure}[h]\centering
		\includegraphics[scale =1]{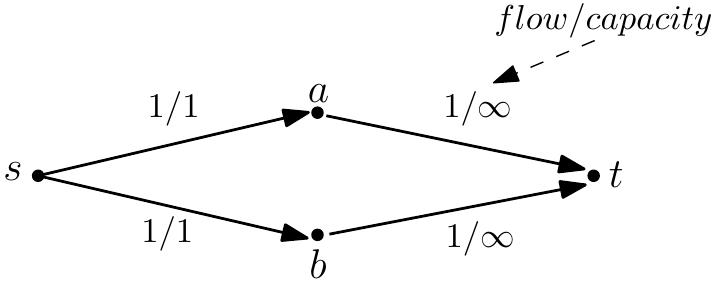}
		\caption{A network $G$ with a $2$-bounded flow $f$.}\label{network_5bounded_flow} \label{inputNetwork}
\end{figure}		

\begin{figure}[h]\centering
		\includegraphics[scale = 1]{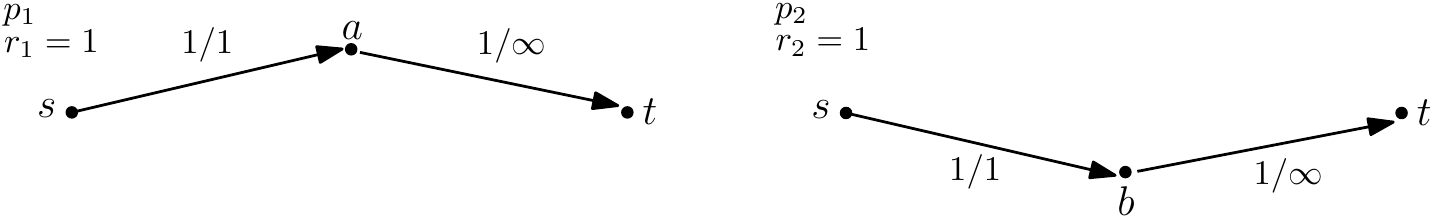}
		\caption{A decomposition of the $2$-bounded flow $f$ into paths 
		$ p_1, p_2$. 
			\label{inputFlow} } 
\end{figure}	

\begin{lemma}[Lemma 4.4 in \cite{koubek1981maximum}]\label{l4.4}
	For each \iks $ T $ (with respect to an \kf ${ f =\{(p_i, r_i): i \in I\}) }$ constructed by the above procedure it holds $ |T| \geq \psi(G)/|I| $.
\end{lemma}

The {\em above procedure} in Lemma \ref{l4.4} refers to a construction of an
increasing $L$-system that is outlined in the original paper. 
As Definition \ref{d4.3} implies $\psi(G) > 0$, we also know by Lemma \ref{l4.4}
that for every increasing $L$-system $T$, $|T| > 0$.

Now we are ready to describe the counter example.
Take $k=2$ and consider the following network $G$ with a $2$-bounded flow $f$
of size $2$ (Fig. \ref{inputNetwork} and \ref{inputFlow}); apparently, this is 
a maximum $2$-bounded flow. 

We are going to show that there exists an increasing system $T$ for $f$. 
According to Lemmas~\ref{l4.2} and \ref{l4.4} this implies the existence 
of a $2$-bounded flow $g$ of size $|f|+|T| > |f|$.
As the flow $f$ is a maximum $2$-bounded flow in $G$, this is a contradiction.

\begin{figure}[ht]\centering
		\includegraphics[scale = 1]{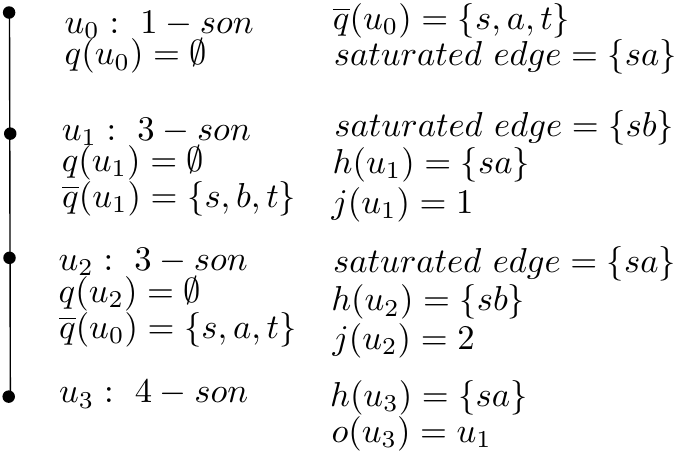}
		\caption{Increasing $ 2 $-system $ T $. 
			\label{increasingTree} }
\end{figure}	

The increasing system $T$ is depicted in Figure \ref{increasingTree}; for the
sake of simplicity, we list only the most relevant attributes.  It is just a
matter of a mechanical effort to check that it meets Definition 4.1 of the
increasing system from the original paper. 

In words, the essence of the counter example is the following.
The purpose of the root of the tree, the node $u_0$, is to increase the flow
from $s$ to $t$ along the path $q(u_0)\bar q(u_0)$ which is (accidently)
the path $p_1$. As there is a saturated edge on this path, namely the edge
$sa$, there is a $3$-son of the node $u_0$, the node $u_1$, whose purpose
is to desaturate the edge $sa$ by diverting one of the paths that use the
edge $sa$ along an alternative route; in particular, the node $u_1$ is
diverting the path $p_1$ and it is diverting it from the very beginning,
from $s$, along the path $q(u_1)\bar q(u_1)$ which is (accidently) the
path $p_2$. 

As there is a saturated edge on this path, namely the edge
$sb$, there is a $3$-son of the node $u_1$, the node $u_2$, whose purpose
is to desaturate the edge $sb$ by diverting one of the paths that use the
edge $sb$ along an alternative route; in particular, the node $u_2$ is
diverting the path $p_2$ and it is diverting it from the very beginning,
from $s$, along the path $q(u_2)\bar q(u_2)$ which is (accidently) 
again the path $p_1$. 

As there is a saturated edge on this path, namely the edge
$sa$, and as there is already another node in the tree that is
desaturating $sa$, namely the node $u_1$, the node $u_2$ does not have
a $3$-son but it has a $4$-son $u_3$ instead, which is a pointer to the 
$3$-son $u_1$. This way, there is a kind of a deadlock cycle in the 
increasing system: $u_1$ is desaturating the edge $sa$ for the node $u_0$ but
it itself needs $u_2$ to desaturate the edge $sb$ in it and $u_2$ in turn needs
$u_3$ to desaturate the edge $sa$, but $u_3$ delegates this task back
to $u_1$. 

At this point, we know that Lemma \ref{l4.2} or Lemma \ref{l4.4} is not correct.
By Definition~\ref{d4.2}, one can check that $|T|=1/2$ which implies,
as we started with a maximum flow, that it is Lemma~\ref{l4.2} that does
not hold.

\section{FPTAS for maximum $L$-bounded flow}\label{chap:approximations}

We first describe a fully polynomial approximation scheme for maximum
$L$-bounded flow on networks with unit edge length. The algorithm is based on 
the algorithm for the maximum multicommodity flow by Garg and 
K\"{o}nemann~\cite{garg2007apx}.

Then we describe a FPTAS for the $L$-bounded
flow problem with general edge lengths.  Our approximation schemas for the
maximum $L$-bounded flow on unit edge lengths and the maximum $L$-bounded flow
with edge lengths are almost identical, the only difference is in using an
approximate subroutine for resource constrained shortest path in the general
case which slightly complicates the analysis.

\subsection{FPTAS for Unit Edge Lengths}\label{sec:apxunit}

Let us consider the path based linear programming (LP) formulation of the maximum $L$-bounded flow, 
\textbf{P}$_{\text{path}}$, and its dual, \textbf{D}$_{\text{path}}$.
We assume that $G=(V,E,c,s,t)$ is a given network and $L$ is a given 
length bound.
Let $\mathcal{P}_L$ denote the set of all $s$-$t$ paths of length at most
$L$ in $G$. There is a primal variable $x(p)$ for each path $p\in \mathcal{P}_L$,
and a dual variable $y(e)$ for each edge $e\in E$.
Note that the dual LP  is a relaxation of an
integer LP formulation of the minimum $L$-bounded cut problem.

\begin{minipage}{.5\linewidth}
  \centering
\begin{equation*}\label{eq:pathlp}
\begin{alignedat}{2}
    &\text{max} \ & \sum_{P\in\mathcal{P}_L}x(P) & \\
    &\text{s.t.} \ & \sum_{\substack{P\in\mathcal{P}_L:\\e\in P}}x(P) &\le
c(e)\quad \forall e\in E\\
    && x &\ge 0
\end{alignedat}
\end{equation*}

\end{minipage}%
\begin{minipage}{.5\linewidth}
  \centering 
\begin{equation*}\label{eq:pathdual}
\begin{alignedat}{2}
    &\text{min} \ & \sum_{e\in E}c(e)y(e) &  \\
    &\text{s.t.} \ & \sum_{e\in P}y(e) &\ge 1\quad \forall P\in
\mathcal{P}_L\\
    && y &\ge 0
\end{alignedat}
\end{equation*}
\end{minipage}

The algorithm simultaneously constructs solutions for the maximum $L$-bounded
flow and the minimum fractional $L$-bounded cut.  It iteratively routes flow
over shortest paths with respect to properly chosen dual edge lengths and at
the same time increases these dual lengths; dual edge length of the edge $e$
after $i$ iterations will be denoted by $y_i(e)$.  During the runtime of the
algorithm, the constructed flow need not respect the edge capacities; however,
with the right choice of parameters $\eps, \delta$ the resulting flow can be
scaled down to a feasible (i.e., respecting the edge capacities) 
flow (Lemma~\ref{lem:feasibleflow}) that is a
$(1+\eps)$-approximation of the maximum $L$-bounded flow
(Theorem~\ref{thm:approxunit}).

For a vector $y$ of dual variables, let $d^L_y(s,t)$ denote the length of the
$y$-shortest $s-t$ path from the set of paths $\mathcal{P}_L$ and let
$\alpha^L(i)=d_{y_i}^L(s,t)$.  Note that a shortest $s-t$ path with respect to
edge lengths $y$ that uses at most a given number of edges can be computed in
polynomial time by a modification of the Dijkstra's shortest path algorithm.

\begin{algorithm}
\caption{\textsc{Approx}($\eps, \delta$)}
\label{algo:approx}
\begin{algorithmic}[1]
\STATE $i\leftarrow 0, \ y_0(e) \leftarrow \delta \quad \forall e\in E,\  x_0(P)\leftarrow 0
\quad \forall P\in \mathcal{P}_L$
\WHILE{$\alpha^L(i)<1$}
\STATE $i\leftarrow i+1$
\STATE $x_i \leftarrow x_{i-1}, y_i \leftarrow y_{i-1}$
\STATE $P \leftarrow y_i$-shortest $s$-$t$ path with at most $L$ edges
\STATE $c \leftarrow \min\limits_{e\in P}c(e)$
\STATE $x_i(P)\leftarrow x_i(P)+c$
\STATE $y_i(e)\leftarrow y_{i}(e)(1+\eps c/c(e)) \quad \forall e\in P$
\ENDWHILE
\RETURN $x_i$
\end{algorithmic}
\end{algorithm}

Let $f_i$ denote the size of the flow after $i$ iterations,
$f_i=\sum_{P\in\mathcal{P}_L}x_i(P)$, and let $\tau$ denote the total number of
iterations performed by \textsc{Approx}; then $x_{\tau}$ is the
output of the algorithm and $f_{\tau}$ its size.

\begin{lemma}
\label{lem:feasibleflow}
The flow $x_{\tau}$ scaled down by a factor of $\log_{1+\eps}\frac{1+\eps}{\delta}$
is a feasible $L$-bounded flow.
\end{lemma}

\begin{proof}
By construction, for every $i$, $x_i$ is an $L$-bounded flow. Thus, we only have
to care about the feasibility of the flow
\begin{equation}
\frac{x_{\tau}}{\log_{1+\eps}\frac{1+\eps}{\delta}}\ .
\end{equation}

For every iteration $i$ and every edge $e\in E$, as $\alpha^L(i-1)<1$, we also
have $y_{i-1}(e) < 1$ and so $y_i(e)< 1+\eps$. It follows that 
\begin{equation}\label{eqn:lengthtupper}
y_{\tau}(e)<1+\eps \ .
\end{equation}

Consider an arbitrary edge $e\in E$ and suppose that the flow $f_{\tau}(e)$
along $e$ has been routed in iterations $i_1, i_2,\dots,i_r$ and the amount of
flow routed in iteration $i_j$ is $c_j$.  Then $f_{\tau}(e)=\sum_{j=1}^r c_j$
and $y_{\tau}(e)=\delta\prod_{j=1}^r (1+\eps c_j/c(e))$. Because each $c_j$ was
chosen such that $c_j\le c(e)$, we have by Bernoulli's inequality that $1+\eps
c_j/c(e)\ge
(1+\eps)^{c_j/c(e)}$ and 
\begin{equation}\label{eqn:lengthtlower}
y_{\tau}(e)\ge\delta\prod_{j=1}^r(1+\eps)^{c_j/c(e)} = \delta
(1+\eps)^{f_{\tau}(e)/c(e)}.
\end{equation}
Combining inequalities (\ref{eqn:lengthtupper}) and (\ref{eqn:lengthtlower}) gives
$$\frac{f_{\tau}(e)}{c(e)}\le\log_{1+\eps}\frac{1+\eps}{\delta}$$ 
which completes the proof.
\end{proof}

\begin{claim} \label{lem:alpha}
For $i=1,\ldots,\tau$,
\begin{align}
\alpha^L(i)& \le\delta L e^{\eps f_i/\beta} \ .
\end{align}
\end{claim}
\begin{proof}
For a vector $y$ of dual variables, 
let $D(y)=\sum_e c(e)y(e)$ and let $\beta=\min_y D(y)/d^L_y(s,t)$. Note that
$\beta$ is equal to the optimal value of the dual linear program.
For notational simplicity we abbreviate $D(y_i)$ as $D(i)$.

Let $P_i$ be the path chosen in iteration $i$ and $c_i$ be the value of $c$ in
iteration $i$. For every $i\ge 1$ we have 
\begin{align}
D(i) 	&= \sum_{e\in E} y_i(e)c(e)\nonumber\\
	&= \sum_{e\in E} y_{i-1}(e)c(e) + \eps   \sum_{e\in P_i} y_{i-1}(e)c_i\nonumber\\
	&= D(i-1) + \eps(f_i-f_{i-1})\alpha^L(i-1)\nonumber
\end{align}
which implies that
\begin{equation}\label{eqn:di}
D(i)= D(0)+\eps\sum_{j=1}^i(f_j-f_{j-1})\alpha^L(j-1).
\end{equation}

Now consider the length function $y_i-y_0$. Note that $D(y_i-y_0)=D(i)-D(0)$
and $d^L_{y_i-y_0}(s,t)\ge \alpha^L(i)-\delta L$.
Hence, 
\begin{equation}
\label{eqn:beta}
\beta\le\frac{D(y_i-y_0)}{d_{y_i-y_0}^L(s,t)}\le\frac{D(i)-D(0)}{\alpha^L(i)-\delta
L} \ . 
\end{equation}
By combining relations~(\ref{eqn:di}) and~(\ref{eqn:beta}) we get
$$\alpha^L(i)\le \delta L
+\frac{\eps}{\beta}\sum_{j=1}^i(f_j-f_{j-1})\alpha^L(j-1) \ .$$

Now we define $z(0)=\alpha^L(0)$ and for $i=1,\ldots,\tau$, $z(i)=\delta L
+\frac{\eps}{\beta}\sum_{j=1}^i(f_j-f_{j-1})z(j-1)$.
Note that for each $i$, $\alpha^L(i) \le z(i)$.  Furthermore,
\begin{align*}
z(i)	&= \delta L +\frac{\eps}{\beta}\sum_{j=1}^i(f_j-f_{j-1})z(j-1)\\ 
	&= \left(\delta L
+\frac{\eps}{\beta}\sum_{j=1}^{i-1}(f_j-f_{j-1})z(j-1)\right)+\frac{\eps}{\beta}(f_i-f_{i-1})z(i-1)\\ 
	&=z(i-1)(1+\eps(f_i-f_{i-1})/\beta)\\
	&\le z(i-1)e^{\eps(f_i-f_{i-1})/\beta}.
\end{align*}
Since $z(0)\le\delta L$, we have $z(i)\le \delta Le^{\eps f_i/\beta}$, and thus also, 
for $i=1,\ldots,\tau$, 
$\alpha^L(i)\le\delta L e^{\eps f_i/\beta}\ .$
\end{proof}

\begin{theorem}\label{thm:approxunit}
For every $0<\eps<1$ there is an algorithm that computes an
$(1+\eps)$-approximation to the maximum $L$-bounded flow in a network with
unit edge lengths in time $\O(\eps^{-2}m^2L\log L)$. 
\end{theorem}
\begin{proof}
We start by showing that for every $\eps<\frac13$ there is a constant 
$\delta=\delta(\epsilon)$ such that $x_{\tau}$, the output of \textsc{Approx}($\eps,\delta$), 
scaled down by $\log_{1+\epsilon}\frac{1+\epsilon}{\delta}$ as in Lemma~\ref{lem:feasibleflow},
is a $(1+3\eps)$-approximation.

Let $\gamma$ denote the approximation ratio of such an algorithm, that is,
let $\gamma$ denote the ratio of the optimal dual solution ($\beta$) to 
the appropriately scaled output of \textsc{Approx}($\eps,\delta$),
\begin{equation}
\gamma  = \frac{\beta \log_{1+\epsilon}\frac{1+\epsilon}{\delta}}{f_{\tau}} \ ,
\end{equation}
where the constant $\delta$ will be specified later.

By Claim~\ref{lem:alpha} and the stopping condition of the while cycle 
we have
\begin{equation}
1\le\alpha^L(\tau)\le\delta L e^{\eps f_{\tau}/\beta}\nonumber
\end{equation}
and hence
$$\frac{\beta}{f_{\tau}}\le\frac{\eps}{\log\frac{1}{\delta L}}.$$
Plugging this bound in the equality for the approximation ratio $\gamma$, we obtain
$$\gamma\le\frac{\eps\log_{1+\eps}\frac{1+\eps}{\delta}}{\log\frac{1}{\delta
L}}=\frac{\eps}{\log(1+\eps)}\frac{\log\frac{1+\eps}{\delta}}{
\log\frac{1}{\delta L}}.$$
Setting $\delta=\frac{1+\eps}{((1+\eps)L)^{1/\eps}}$ yields
$$\frac{\log\frac{1+\eps}{\delta}}{
\log\frac{1}{\delta L}} = \frac{\frac{1}{\eps}\log((1+\eps)L)}{
\left(\frac1\eps-1\right)\log((1+\eps)L)}=\frac1{1-\eps}.$$
Taylor expansion of $\log(1+\eps)$ gives a bound
$\log(1+\eps)\ge\eps-\frac{\eps^2}{2}$ for $\eps<1$ and it follows for
$\eps<\frac13$ that
$$\gamma\le\frac{\eps}{(1-\eps)\log(1+\eps)}\le
\frac{\eps}{(1-\eps)(\eps-\eps^2/2)}\le\frac{1}{1-\frac32\eps}\le 1+3\eps.$$

To complete the proof, we just put $\eps'=\eps/3$ and run 
\textsc{Approx}($\eps',\delta(\eps')$).
It remains to prove the time complexity of the algorithm.
In every iteration $i$ of \textsc{Approx}, the length $y_i(e)$ of an edge $e$
with the smallest capacity on the chosen path $P$ is increased by a factor of
$1+\eps'$.  Because $P$ was chosen such that $y_i(P)<1$ also $y_i(e)<1$ for
every edge $e \in P$. Lengths of other edges get increased by a factor
of at most $1+\eps'$, therefore $y_{\tau}(e)<1+\eps'$ for every edge $e\in E$.
Every edge has the minimum capacity on the chosen path in at most
$\left\lceil\log_{1+\eps'}\frac{1+\eps'}{\delta} \right\rceil =
\O(\frac1\eps\log_{1+\eps}L)$ iterations, so \textsc{Approx} makes at most
$\O(\frac m{\eps}\log_{1+\eps}L)=\O(\frac{m}{\eps^2}\log L)$ iterations. 

Each iteration takes time $\O(Lm)$ so the total time taken by \textsc{Approx} is
$\O(\eps^{-2}m^2L\log L)$. 
\end{proof}

\subsection{FPTAS for General Edge Lengths}\label{sec:apxgeneral}

Now we extend the approximation algorithm to networks with general
edge lengths that are given by  a length function $\ell:E\to \mathbb{N}$. 
The dynamic programming algorithm for computing shortest paths
that have a restricted length with respect to another length function,
does not work in this case. In fact, the problem of finding shortest path with
respect to a given edge length function while restricting to paths of bounded
length with respect to another length function is \cNP-hard in general~\cite{HZ:80}.
On the other hand, there exists a FPTAS for it~\cite{hassin1992restricted,lorenz2001restricted}.

We assume that we are given as a black-box an algorithm that for a given
graph $G$, two edge length functions $y$ and $\ell$, two distinguished
vertices $s$ and $t$ from $G$, a length bound $L$ and an error parameter $w>0$,
computes a $(1+w)$-approximation of the $y$-shortest path of $\ell$-length
at most $L$; we denote by $d_{y,\ell}^L(s,t;w)$ the length of such a path
and we also introduce an abbreviation $\bar\alpha^L(i)=d_{y_i,\ell}^L(s,t;w)$.
Note that for every $i$, $\bar\alpha^L(i) \leq (1+w)\alpha^L(i)$.
We can use the FPTAS of Lorenz and Raz~\cite{lorenz2001restricted} for this
task. 

The algorithm of Garg and K\"{o}nemann~\cite{garg2007apx} for approximating
maximal multicommodity flow has been improved by Fleischer
~\cite{fleischer2000apx}. The original algorithm computes the shortest path
between every terminal pairs in every iteration. Fleischer divided the
algorithm to phases where she worked with commodities one by one. This way her
algorithm effectively works with approximations of shortest paths while
eliminates the dependency on the number of commodities and still gets a good
approximation ratio.
Using a similar analysis we show that we can work with an
approximation shortest path algorithm to get an FPTAS to otherwise intractable maximum 
$L$-bounded flow problem with general edge lengths.

The structure of the $L$-bounded flow algorithm with general edge lengths stays the same as in the unit edge lengths case.
The only difference is that instead of $y$-shortest $L$-bounded paths, 
approximations of $y$-shortest $L$-bounded paths are used (steps 2 and 5).
\begin{algorithm}
\caption{\textsc{ApproxGeneral}($\eps, \delta, w$)}
\label{algo:approxgeneral}
\begin{algorithmic}[1]
\STATE $i\leftarrow 0, \ y_0(e) \leftarrow \delta \quad \forall e\in E, \ x_0(P)\leftarrow 0 \quad \forall P\in \mathcal{P}_L$
\WHILE{$\bar\alpha^L(i)<1+w$}
\STATE $i\leftarrow i+1$
\STATE $x_i \leftarrow x_{i-1}, y_i \leftarrow y_{i-1}$
\STATE $P \leftarrow (1+w)$-approximation of the $y_i$-shortest $L$-bounded path
\STATE $c \leftarrow \min\limits_{e\in P}c(e)$
\STATE $x_i(P)\leftarrow x_i(P)+c$
\STATE $y_i(e)\leftarrow y_{i}(e)(1+\eps c/c(e)) \quad \forall e\in P$
\ENDWHILE
\RETURN $x_i$
\end{algorithmic}
\end{algorithm}

The analysis of the algorithm follows the same steps as the analysis of
Algorithm~\ref{algo:approx} but one has to be more careful when dealing with
the lengths.

As in the previous subsection, let $f_i$ denote the size of the
flow after $i$ iterations 
and let $\tau$
denote the total number of iterations.
Due to the lack of space, the proofs are given in the Appendix.

\begin{lemma}
\label{lem:feasibleflowgeneral}
The flow $x_{\tau}$ scaled down by a factor of $\log_{1+\eps}\frac{(1+\eps)(1+w)}{\delta}$
is a feasible $L$-bounded flow.
\end{lemma}

\begin{proof}
For every edge $e\in E$ and iteration $i$, as $\bar \alpha^L(i-1)<1+w$, we 
also have $y_{i-1}(e) < 1+w$. By description of the algorithms, this implies 
$y_i(e)< (1+\eps)(1+w)$, and in particular,
\begin{equation}\label{eqn:lengthtuppergeneral}
y_{\tau}(e)<(1+\eps)(1+w) \ .
\end{equation}
Combining this with $y_{\tau}(e)\ge\delta (1+\eps)^{f_{\tau}(e)/c(e)}$
from inequality~(\ref{eqn:lengthtlower}) in previous subsection,
we derive
$$\frac{f_{\tau}(e)}{c(e)}\le\log_{1+\eps}\frac{(1+\eps)(1+w)}{\delta}$$
which completes the proof.
\end{proof}

\begin{claim} \label{lem:alpha2}
For $i=1,\ldots,\tau$,
\begin{align}
\alpha^L(i)& \le\delta L e^{\eps (1+w)f_i/\beta} \ .
\end{align}
\end{claim}
\begin{proof}
By the same reasoning as in the proof of Claim~\ref{lem:alpha}, we obtain
\begin{equation}\label{eqn:digeneral}
D(i) \le D(0)+\eps\sum_{j=1}^i(f_j-f_{j-1})(1+w)\alpha^L(i-1) \ ,
\end{equation}
where the extra $1+w$ factors stems from the fact that we work, in iteration $i$, 
not with a path of length $\alpha(i)$ but with a path of length 
$\bar\alpha(i)\leq(1+w)\alpha(i)$.
Combining this with $\beta\le\frac{D(i)-D(0)}{\alpha^L(i)-\delta
L}$ from inequality~(\ref{eqn:beta}), we obtain
$$\alpha^L(i)\le \delta L +\frac{\eps(1+w)}{\beta}\sum_{j=1}^i(f_j-f_{j-1})\alpha^L(j-1) \ .$$

From this point, we proceed again along the same lines as in the proof of 
Claim~\ref{lem:alpha} (the only difference is that instead of $\epsilon/\beta$,
we work now with $(1+w)\epsilon/\beta$) and get the desired bound.
\end{proof}

\begin{theorem} \label{thm:gen}
There is an algorithm that computes an $(1+\eps)$-approximation to the
maximum $L$-bounded flow in a graph with general edge lengths in time
$\O(\frac{m^2n}{\eps^2}\log L(\log\log n + \frac{1}{\eps}))$. 
\end{theorem}
\begin{proof}
We show that for every $\eps\le\frac13$ there are constants
$\delta$ and~$w$ such that $x_{\tau}$, the output of
\textsc{ApproxGeneral}($\eps, \delta, w$), scaled down by
$\log_{1+\eps}\frac{(1+\eps)(1+w)}{\delta}$ as in
Lemma~\ref{lem:feasibleflowgeneral}, is a $(1+5\eps)$-approximation to the
maximum $L$-bounded flow with general capacities; the theorem easily
follows. 

Let $\gamma$ denote the approximation ratio of such an algorithm, that is,
let $\gamma$ denote the ratio of the optimal dual solution ($\beta$) to
the appropriately scaled output of \textsc{ApproxGeneral}($\eps,\delta,w$),
\begin{equation}
\gamma  = \frac{\beta \log_{1+\epsilon}\frac{(1+\epsilon)(1+w)}{\delta}}{f_{\tau}} \ ,
\end{equation}
where the constants $\delta$ and $w$ will be specified later.

By the stopping condition of the while cycle we have 
$1+w \leq \bar\alpha^L(\tau) \leq (1+w)\alpha^L(\tau)$, that is, $1\leq \alpha^L(\tau)$; 
combining it with Claim~\ref{lem:alpha2}, we get
$$\frac{\beta}{f_{\tau}}\le\frac{\eps(1+w)}{\log\frac{1}{\delta L }}.$$
Plugging this bound in the equality for the approximation ratio $\gamma$,
we obtain
\begin{equation} \label{eqn:gamma}
\gamma\le
\frac{\eps(1+w)\log_{1+\eps}\frac{(1+\eps)(1+w)}{\delta}}{\log\frac{1}{\delta
L}}=\frac{\eps(1+w)}{\log(1+\eps)}\frac{\log\frac{(1+\eps)(1+w)}{\delta}}{
\log\frac{1}{\delta L}}\ .
\end{equation}
Setting $\delta=\frac{(1+\eps)(1+w)}{((1+\eps)(1+w)L)^{1/\eps}}$
yields
\begin{equation}
\frac{\log\frac{(1+\eps)(1+w)}{\delta}}{
\log\frac{1}{\delta L}} = \frac{\frac{1}{\eps}\log((1+\eps)(1+w)L)}{
\left(\frac1\eps-1\right)\log((1+\eps)(1+w)L)}=\frac1{1-\eps}\ .
\end{equation}
Thus, the bound on the approximation ratio $\gamma$~(\ref{eqn:gamma}) simplifies to 
$$\gamma\le\frac{\eps(1+w)}{(1-\eps)\log(1+\eps)}\le
\frac{\eps(1+w)}{(1-\eps)(\eps-\frac{\eps^2}{2})}\le
\frac{1+w}{1-\frac{3}{2}\eps}\ ,$$
where the second inequality follows from the Taylor expansion of $\log(1+\eps)$ 
and the bound $\log(1+\eps)\ge\eps-\frac{\eps^2}{2}$, for $\eps<1$. 
By setting $w=\eps$, for $\eps\le\frac{1}{3}$ we get the promised bound
$$\gamma\le\frac{1+w}{1-\frac{3}{2}\eps}\leq (1+\eps)(1+3\eps)\leq1+5\eps \ .$$

Concerning the running time, we observe that in
every iteration the length of at least one edge gets increased by the ratio
$1+\eps$. For every edge $e\in E$ we have $y_{\tau}(e)\le(1+\eps)(1+w)$. By
the same arguments as in the previous subsection, our choice of the parameters
ensures that the total number of iterations is at most
$\O(\frac{m}{\eps}\log_{1+\eps}L)=\O(\frac{m}{\eps^2}\log L)$.
The FPTAS approximating the resource bounded shortest path takes time
$\O(mn(\log\log n+\frac{1}{\eps}))$. Combining these two bounds completes 
the proof.
\end{proof}

We note that the exponential length method can be used for many fractional
packing problems and using the same technique we could get an approximation
algorithm for maximum multicommodity $L$-bounded flow.

\bibliographystyle{abbrv}
\bibliography{bibliography}

\end{document}

%% file: makra_full.tex
\theoremstyle{plain}
\newtheorem{thm}{Theorem}
\newtheorem{lemma}[thm]{Lemma}
\newtheorem{claim}[thm]{Claim}

\newtheorem{theorem}{Theorem}

\theoremstyle{plain}
\newtheorem{defn}{Definition}

\theoremstyle{remark}

\newcommand{\iks}{increasing $ L $-system }

\newcommand{\kf}{$ L $-flow }

\newcommand{\kc}{$ L $-cut }
\newcommand{\kbf}{$ L $-bounded flow }
\newcommand{\kbc}{$ L $-bounded cut }

\newcommand{\kfw}{$ L $-flow}
\newcommand{\kffw}{$ L $-flow $f$}
\newcommand{\kcw}{$ L $-cut}
\newcommand{\kbfw}{$ L $-bounded flow}
\newcommand{\kbcw}{$ L $-bounded cut}

\def\eps{\varepsilon}
\def\O{\mathcal{O}{}}

\def\cNP{\hbox{\rm \sffamily NP}}

